\documentclass[12pt]{article}

\usepackage{epsfig}

\usepackage{amssymb}
\usepackage{amsfonts}

\usepackage{color}
 
\def\be{\begin{equation}}
\def\ee{\end{equation}}
\def\ba{\begin{array}{c}}
\def\ea{\end{array}}

\def\ben{$$}
\def\een{$$}

\newcommand{\bea}{\begin{eqnarray}}
\newcommand{\eea}{\end{eqnarray}}

\newcommand{\kt}{\rangle}
\newcommand{\br}{\langle}

\newtheorem{thm}{Theorem}

\newtheorem{prop}[thm]{Proposition}

\newtheorem{conj}[thm]{Conjecture}
\newenvironment{proof}{\noindent
 {\bf Proof.}}{\hfill$\square$\vspace{3mm}\endtrivlist}

\begin{document}

\begin{center}

{\Large

Confluences of exceptional points
and a systematic
classification of
quantum catastrophes

}

\vspace{10mm}

\textbf{Miloslav Znojil}

\vspace{0.2cm}

The Czech Academy of Sciences, Nuclear Physics Institute,

 Hlavn\'{\i} 130,
250 68 \v{R}e\v{z}, Czech Republic

\vspace{0.2cm}

 and

\vspace{0.2cm}

Department of Physics, Faculty of Science, University of Hradec
Kr\'{a}lov\'{e},

Rokitansk\'{e}ho 62, 50003 Hradec Kr\'{a}lov\'{e},
 Czech Republic

\vspace{0.2cm}

 and

\vspace{0.2cm}

Institute of System Science, Durban University of Technology,
P. O. Box 1334, Durban, 4000, South Africa

\vspace{0.2cm}

{e-mail: znojil@ujf.cas.cz}

\end{center}

\vspace{15mm}

\newpage

\section*{Abstract}

In the problem of classification of the
parameter-controlled quantum phase transitions,
attention is turned from the
conventional manipulations
with the energy-level mergers
at exceptional points
to the control of
mergers of the exceptional points
themselves.
What is obtained is an exhaustive
classification
which characterizes
every phase transition by the algebraic
and geometric multiplicity
of the underlying confluent exceptional point.
Typical qualitative characteristics of
non-equivalent phase transitions are illustrated
via a few elementary toy models.

\section*{Keywords}

.

quasi-Hermitian quantum Hamiltonians;

loss of observability processes;

exceptional point multiplicities;

phase transition classification.

\newpage

\section{Introduction}

 \noindent
In the conventional descriptions of unitary evolution
of quantum systems in
Schr\"{o}dinger picture (SP, \cite{Messiah})
the information about dynamics is all carried by the
Hamiltonian. The predictions of experiments are all based on
the
solution of
Schr\"{o}dinger equation
 \be
 {\rm i}\frac{d}{d t}|\psi(t)\kt=H\,
 |\psi(t)\kt\, ,\ \ \
  \ \ \ \ |\psi(t)\kt \in {\cal V}\,.
 \label{eqsch}
 \ee
A broader perception of the concept
of quantum dynamics
will be advocated here,
with emphasis upon
{\em qualitative\,} aspects of the role of parameters.

Among several sources of inspiration of our project the most
obvious one is the Thom's theory of
catastrophes \cite{Zeeman,Zeema}. As long as such a theory is
of a purely geometric nature \cite{Zeemanb}, it is applicable,
predominantly, to the classical dynamical systems. In this domain it
offers, first of all, a systematic classification of structures and
of the possible changes of structures of the long-term classical
equilibria \cite{Zeemanc}. The applicability and/or an immediate
transfer of these ideas to quantum systems are limited
\cite{ODell,ODellb,ODellc,catast,ODelld}.

According to
papers \cite{Heissb,Heiss,Heissc,hata,juta,kajuta,ceta},
new eligible directions of research emerged after
a turn of attention to the Kato's
concept of exceptional point
(EP, \cite{Kato}).
After a small change $g \to g^{(EP)}$ of
a real or complex parameter
in Hamiltonian $H=H(g)$,
such an operator ceases to be diagonalizable.
Hence,
the bifurcation of the Thom's classical equilibria
can find its genuine quantum analogue
in the passage of the parameter through its
real or complex
value $g^{(EP)}$.
The latter possibility is, in essence, also the key point of our
present paper. In our text we will try to develop the idea in a
certain more systematic and constructive manner.

A decisive key to the realizability of the project can be seen in
the Bender's and Boettcher's change of the traditional paradigms
\cite{BB}. Indeed, it was them who conjectured that the unitary
evolution could be, under certain conditions, realized and described
even when the generator $H$ of evolution of the wave function in
Schr\"{o}dinger Eq.~(\ref{eqsch}) becomes, in an apparent
contradiction to the well known Stone theorem \cite{Stone},
manifestly non-Hermitian. And precisely this change of paradigm (cf.
also the detailed outline of the resulting consistent quantum theory
of closed systems as reviewed, more than ten years ago, in papers
\cite{Carl,ali}) opened also the way towards the change of the
status of the notion of EPs from a strictly mathematical tool as
developed in the Kato's book to one of the most
important concepts in experimental physics -- see, e.g., paper
\cite{[j]} outlining the related
``roadmap for future studies and potential
applications''.

The core of our present message will lie in a combination of the
purposeful theoretical use of parameter-dependent non-Hermitian
Hamiltonians with a detailed analysis of some of the consequences in
the ambitious phenomenological context of description
and classification
of a broad
class of phenomena called quantum phase transitions \cite{[a]}.
Naturally, the
feasibility of our project will require a certain methodically
motivated narrowing of its scope.
Thus, in contrast to the more conventional perception of the quantum
phase transition phenomena as described in textbook \cite{[a]} (and,
typically, associated with the spontaneous symmetry breaking), our
present approach to the problem of phases will be slightly
different, more closely associated with the potential passage of the
quantum system in question strictly through its EP singularity.

In the context of such a theory
(rendered consistent by the non-Hermiticity of $H$)
we will only consider a number of elementary toy models, with
emphasis upon the quick, non-numerical solvability of the related
Schr\"{o}dinger Eq.~(\ref{eqsch}). It
is worth adding that in such a case (to be related here, for the
sake of simplicity, just to the spectral phase transitions) one of
the phases may (though need not) be ill-defined in a way depending
on the respective presence or absence of the complexification of the
spectrum near the EP (see, e.g., \cite{procA} for a few
illustrative examples of the latter, slightly less well known
possibility of having no complexification).

In the currently highly popular pragmatic context of the
phenomenological
applicabilty and of the
proposals and predictions of the results of experiments,
the unavoidable
methodical limitations of
our present considerations
using oversimplified toy models
will be even more visible and
restrictive. In this
respect the readers may
be recommended to fill the experiment-related gaps in our text by
following the currently existing and rich specialized literature (see,
e.g., the freshmost reviews of non-Hermitian physics
in \cite{Christodoulides,Carlbook}).

In
the latter frame we will only emphasize the central role played, in the
underlying mathematics and physics, by the ubiquitous \cite{Heiss}
Kato's notion of EPs.
In the language of mathematics we only intend to complement the
contemporary popular but rather formal reference to EPs in various
realistic models by a slightly more ambitious theoretical
interpretation of the EP concept referring to its non-equivalent
realizations.

\section{Unitarity-of-evolution constraint}

Among the existing applications
of qualitative considerations
to quantum dynamics we felt particularly addressed
by the
mathematical studies in which the
EP limits were of
order two (EP2).
In this scenario,
just some
two neighboring eigenvalues
$E_n(g)$ and $E_{n+1}(g)$ of $H(g)$
are assumed to
merge
and complexify
at
$g = g^{(EP)}= g^{(EP2)}$.

The latter studies were
often motivated by
the physics of systems exhibiting
a genuine quantum phase
transition \cite{Heissmod,issmod}.
According to
our most recent commentary \cite{confser},
most of these systems have been considered
``open'',
interacting with a
certain not too well specified ``environment''.
As a consequence,
the bound states remained unstable, with the
energies
which need not be kept real \cite{Nimrod}.
In such an open-system setup
a typical Hamiltonian $H(g)$ is non-Hermitian
so that its EP singularities of the $N-$th-order
may be complex,
$g= g^{(EPN)} \in \mathbb{C}$.
One can, nevertheless, hardly speak about fundamental theory
because the ``input'' information about the
open-system dynamics
(and, in particular, about the environment)
remains incomplete.

Our present attention will be restricted
to the closed systems
characterized by the unitarity of their
evolution.
One of the key technical consequences
is that
the postulate of unitarity
must be, due to the Stone theorem \cite{Stone},
necessarily
connected with the postulate of Hermiticity
of the Hamiltonian.

The way out of the apparent impasse
has only been discovered very recently.
It appeared sufficient to
replace the
conventional textbook SP paradigm
by its straightforward upgrade
which works with the {\em two\,}
non-equivalent Hermitian conjugations
and
which may be called
pseudo-Hermitian quantum mechanics (PHQM,
see its review \cite{ali}).

\subsection{Quantum observables in pseudo-Hermitian representation.}

In the PHQM SP approach
the
EP singularity may
mark a natural
boundary of the formal
acceptability of any candidate for quantum
Hamiltonian. The theory emphasizes that the mere
specification of the linear space ${\cal V}$
and the related knowledge of
the ket-vector solutions $|\psi(t)\kt \in {\cal V}$ of
Schr\"{o}dinger Eq.~(\ref{eqsch}) are insufficient.
What is considered equally important is the freedom
of the choice of the physical inner product between states.
This is equivalent to
the specification of a correct dual space
${\cal V}'$
of the linear functionals in ${\cal V}$.
Such a choice
is known to be ambiguous
(see, e.g., p. 246 in \cite{Messiah}).
Still, in contrast to the widespread beliefs,
this ambiguity may be useful, bringing several
immediate theoretical challenges as well
as practical benefits~\cite{arabky}.

Some of the subtler aspects
of the problem did not find their ultimate clarification
yet \cite{lotor}.
Nevertheless, under certain additional technical
assumptions
the apparent paradox
has already been resolved,
almost thirty years ago,
in review paper \cite{Geyer}.
The
ambiguity of the abstract theory,
i.e., the ambiguity of the choice of the correct physical Hilbert space
${\cal H}_{\rm (physical)}=[{\cal V},{\cal V}'_{\rm (physical)}]$
has been shown removable. It has been explained that
there exists a very natural
method of the necessary unique specification
of the correct
antilinear duality map
${\cal T}_{\rm (physical)}:\,{\cal V}\ \to \ {\cal V}'_{\rm (physical)}$.

In the literature
one still encounters a few
obstinate terminological misunderstandings.
One of their sources lies in the fact
that the
operator ${\cal T}_{\rm (physical)}$
of the correct Hermitian conjugation
need not necessarily have an easily obtainable realization
(cf., e.g., \cite{117,117i,117b,117c}).
The
reconstruction of the physical inner-product space
${\cal H}_{\rm (physical)}$
is,
therefore, most often postponed till the very end of
the calculations.
Temporarily, the correct physical space
is being replaced by its simplified,
user-friendlier alternative
${\cal H}_{\rm (auxiliary)}$.
In spite of being manifestly unphysical,
the key advantage of the latter choice lies
in the simplification of the conjugation.
Its
most straightforward form
${\cal T}_{\rm (auxiliary)}:\,{\cal V}\ \to \ {\cal V}'_{\rm (auxiliary)}$
is realized as the
action which transforms
the column-vector {\it alias\,} ket-vector
$|\psi\kt \in {\cal V}$
into its conventional ``Dirac's''
conjugate of textbooks, i.e.,
into
its
bra-vector partner
$\br \psi| \in {\cal V}'_{\rm (auxiliary)}$ which is
constructed
as a row-vector
composed of the complex-conjugate elements of
its partner $|\psi\kt$.

For the users of the PHQM SP theory it is sufficient to know that
the decisive simplification of it applications is
achieved
via a consequent
representation of all of the states in
${\cal H}_{(auxiliary)}$
rather than in ${\cal H}_{(physical)}$.
The {\em only\,} space in which
one performs calculations is
${\cal H}_{(auxiliary)}$.
Hence, the use of the Dirac's
bra-ket notation conventions
cannot lead to any contradictions.
Under this convention
it is easy to evaluate
any correct inner product $(\psi_1,\psi_2)_{(physical)}$
in ${\cal H}_{(physical)}$
in
terms of its unphysical partner in
${\cal H}_{(auxiliary)}$ because
we are allowed to abbreviate
$(\psi_1,\psi_2)_{(auxiliary)}=\br \psi_1|\psi_2\kt$.
The representation
of the amended inner product
$(\psi_1,\psi_2)_{(physical)}$
remains
based on the definition
$(\psi_1,\psi_2)_{(physical)}=\br \psi_1|\Theta|\psi_2\kt$
in which the new symbol $\Theta$ (called Hilbert space metric)
can in fact
carry a nontrivial part of the information about dynamics.

The picture of reality remains internally consistent.
Whenever one considers a parameter-dependent (and, say, analytic)
family of SP Hamiltonians $H(g)$ admitting
an EP singularity  at a (real or complex) EP value $g=g^{(EP)}$,
one has to localize the domain ${\cal D}_{(physical)}$
of admissible parameter(s) $g$ at which the spectrum
remains real and discrete, i.e.,
unitarity- and closed-system-compatible.
The theory is completed when one specifies also the
Hilbert-space metric
which is, in general,
$g-$dependent,
$\Theta=\Theta(g)$.

The necessary {\em mathematical\,} properties of
the metric operator
can be found thoroughly discussed
in \cite{ali,lotor,Geyer} and in \cite{book}.
Among these properties a key role is played by the
ambiguity of the assignment of the metric
to a preselected SP Hamiltonian
symbolized, whenever needed, by the introduction of
another formal parameter
$c\,$ in $\Theta=\Theta(g,c)$.

Irrespectively of the latter ambiguity,
{\em any\,} Hamiltonian-compatible metric
will necessarily {\em cease to exist\,}
in
the EP limit \cite{SIGMAdva}.
In parallel, our operator $H(g)$
will
cease to be diagonalizable and it will
lose its status of an acceptable Hamiltonian
in {\em the same\,} limit of $g\to g^{(EP)}$.

\subsection{Quantum phase transitions at exceptional points.}

Several well known quantum
effects can be connected
with some EP singularities.
The limit of $g\to g^{(EP)}$
implies the
end (or at least interruption) of the observability
of the quantum system.
In such a limit, typically (cf., e.g.,
the schematic model in \cite{ptsqw}),
at least one pair
of energy levels merges and complexifies, i.e.,
the system ceases to be unitary.
Besides the schematic models
there also exist multiple entirely realistic samples
of such a phenomenon. The best known ones
are encountered in relativistic quantum mechanics.
The emergence of the singularity
requires there an abrupt redefinition of the Hamiltonian
in which one has to incorporate
the new, ``unfrozen'' dynamical degrees of freedom.

The necessary matching
of the old (= ``before EP'') and
new (= ``after EP'') dynamics (i.e., between the
respective {\it ad hoc\,}
Hamiltonians) is usually performed on a
pragmatic, effective-Hamiltonian basis.
The realization of the transition becomes less counterintuitive
when the EP-caused loss of the observability happens to involve
more than two levels. One of the most characteristic illustrative
examples is the well known Landau's \cite{Landau}
strongly singular harmonic oscillator with potential
 \be
 V^{(HO)}(x)=x^2-g/x^2\,.
 \label{land}
 \ee
The system collapses, in suitable
units, at $g=1/4$. One of the ways towards the resolution of the
puzzle has been described in Ref.~\cite{ptho}. At $x=0$ we
regularized the potential in the spirit of pseudo-Hermitian
quantum theory. The collapse of the spectrum then
acquired an immediate EP-related form. With the growth of
attraction $g$
the levels were found to merge and to
form, subsequently, the complex conjugate pairs
(cf. Fig.~\ref{pictho}).

\begin{figure}[h]                    
\begin{center}                         
\epsfig{file=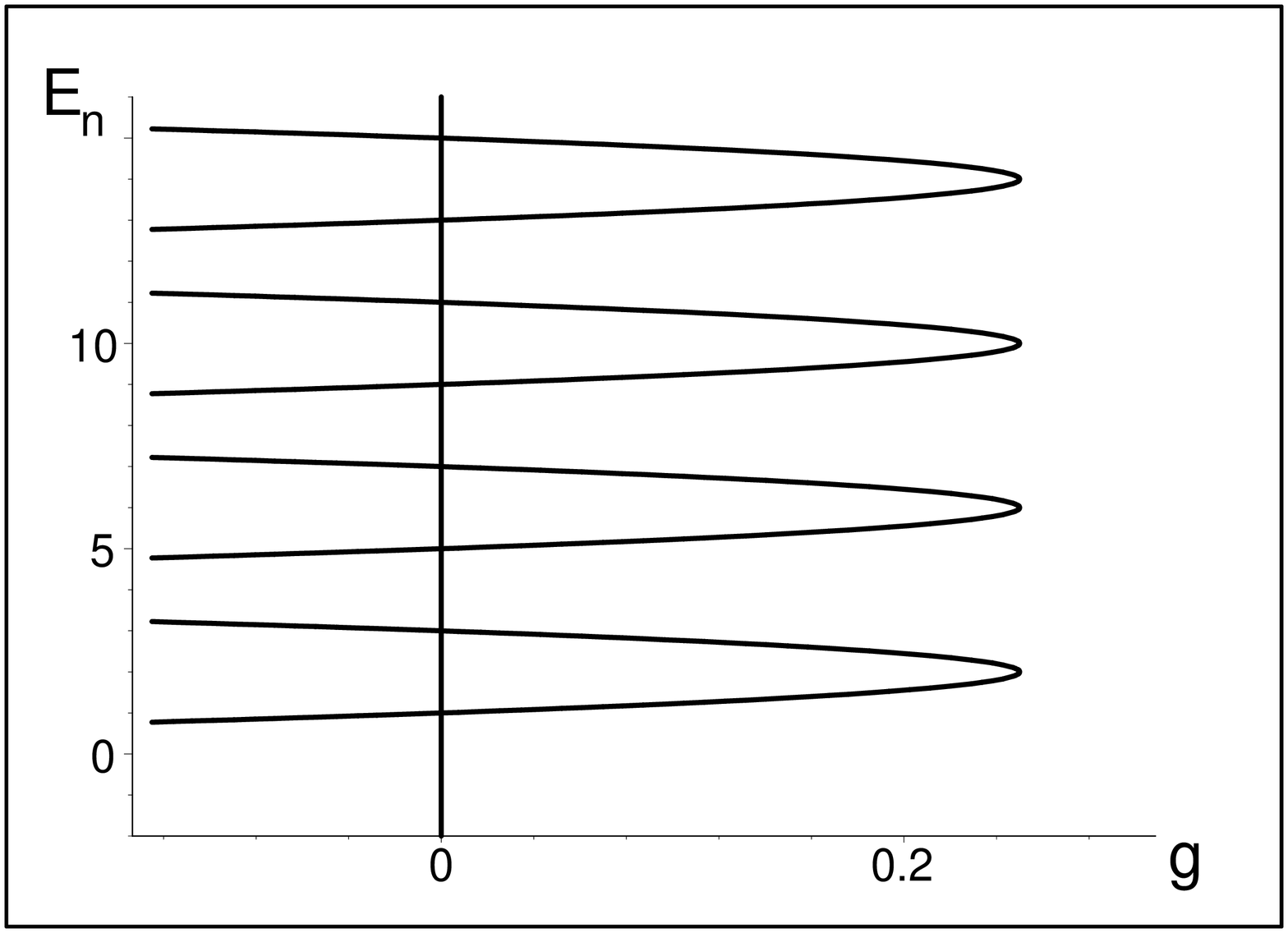,angle=270,width=0.4\textwidth}
\end{center}    
\vspace{2mm} \caption{The low-lying part of the
bound-state spectrum of the exactly solvable
quasi-Hermitian
harmonic oscillator
of Eq.~(\ref{land}).
The model possesses
the confluent exceptional point tractable as the loss-of-unitarity
quantum phase transition at $g=1/4$.
 \label{pictho}
 }
\end{figure}

In our recent follow-up paper \cite{scirep} a
closer connection has been
established between the harmonic-oscillator
physics of collapse and the
mathematics of its exceptional points. Near $g=1/4$,
in particular,
explicit form has been found of all of the admissible
duality maps ${\cal T}$ defining all of the available physical
Hilbert spaces and metrics $\Theta(g,c)$.
With an auxiliary regularization shift
$\varepsilon>0$ of coordinates $x \to x-{\rm i} \varepsilon$
in (\ref{land})
(which
does not influence the results and can be arbitrary) the HO model
has been finally shown to support the desirable
degeneracies of our present interest,
 \be
 \lim_{{{g}}\to {{g}}^{(EP2)}_{m}}
 E_{2m}({{g}})=
 \lim_{{{g}}\to {{g}}^{(EP2)}_{m}}
 E_{2m+1}({{g}})=
 E^{(EP2)}_m<E^{(EP2)}_{m+1}
 \,,\ \ \ \ m = 0,1,\ldots\,.
 \label{dekq}
 \ee
From the point of view of mathematics it is utterly nontrivial that
all of these exceptional-point
couplings {\em coincide},
 \be
 {g}^{(EP2)}_{0}=
  {g}^{(EP2)}_{1}=
{g}^{(EP2)}_{2}= \ldots=
 {g}^{(EP)}_{(confluent)} \,
 \label{negrgr}
 \ee
(see, once more, Fig.~\ref{pictho}).
By far the most interesting physics behind
the latter ``degeneracy of degeneracies''
occurs in a small vicinity of the confluent EP
value. At the slightly weaker couplings ${g} < 1/4$ the {\em whole\,}
spectrum is non-degenerate and real, i.e., the system is
unitary. Whenever we choose just a slightly
stronger attraction ${g} > 1/4$,
the reality of {all of the individual energy
levels} gets simultaneously lost.

For a complementary, qualitatively different introductory
illustration let us
recall
the exactly solvable square-well model $H^{(SQW)}({g})$
of Ref.~\cite{ptsqw}.
The process of the loss of the observability
starts there at the two
lowermost bound states.
With the growth of $g$ (i.e.,
of the non-Hermiticity),
the model produces an infinite
sequence
of the energy mergers of order two (EP2).
They are
real, well separated and
ordered as follows,
 \be
 0<{g}^{(EP2-SQW)}_{0}< {g}^{(EP2-SQW)}_{1}<
 {g}^{(EP2-SQW)}_{2}< \ldots\,.
 \label{grgr}
 \ee
This is a typical, generic scenario.
In many other quantum models with a variable parameter
(cf., e.g., \cite{ptsqwb,ptsqwbc,ptsqwc}),
the EP-caused
quantum-phase-transition phenomenon just exclusively
involves the pairs of the merging levels.
Obviously, the gradually emerging
EP2s are characterized by their nonzero distance
from the parametric domain ${\cal D}_{(physical)}$
so that they remain
phenomenologically irrelevant.
At the same time, their confluence as sampled by
Fig.~\ref{pictho} has always been considered
improbable and next to impossible to
achieve in the laboratory \cite{Heiss}.

A decisive return to optimism as sampled by the theoretical results
of Ref.~\cite{procA} is only of a very recent date.
One has to emphasize that
also in the parallel context of the possible experimental
simulations the recent progress is quick. In \cite{[b]}, for
example, the authors argued that the models as sampled, up to some
similarity transformations, in Ref.~\cite{procA} could really find an
immediate
experimental realization. In detail these authors have shown that
the matrices controlling the evolution of the higher-order field
moments of certain two-mode systems could be realizable in the
zero-dimensional bosonic anti-${\cal PT}$-symmetric dimers.

\subsection{The mergers of the mergers.}

In our present paper we will
ignore the isolated EP2 mergers
of a single pair of energies
 $$
 \lim_{{{g}}\to {{g}}^{(EP2)}}
 [E_{j}({{g}})-
  E_{j+1}({{g}})]=0
 $$
occurring at a single
excitation $j$ as not too interesting.
Our search will be
redirected to the
models exhibiting certain
``mergers of the mergers''.
More explicitly, we will
introduce at least one other variable parameter
(say, $p$)  and we
will search for the confluence of the
EP2s themselves, i.e., of ${g}^{(EP2)}_{a}(p)$
with ${g}^{(EP2)}_{b}(p)$, etc.
Thus, in an ``upgraded'' dynamical scenario
we will search, say,
for four-level merger EP4=EP2$\oplus$EP2 such that
  $$
  \lim_{{{p}}\to {{p}}^{(EP4)}}
 [{g}^{(EP2)}_{a}(p)- {g}^{(EP2)}_{b}(p)]=0
 $$
etc.

In the framework of such a project,
both of our previous illustrative models
proved unsatisfactory.
In the former, harmonic-oscillator case,
the ``merger of all mergers'' did occur but it
remained rigid, parameter-independent, i.e.,
not usable for any active control of dynamics.
In the other, SQW model,
what remained rigid was
the separation of the
exceptional points.
The absence of any auxiliary parameter $p$
did not allow us to convert
at least some of the sharp inequalities
into equal signs
in Eq.~(\ref{grgr}).

An encouraging partial resolution of
the puzzle only came with paper~\cite{procA}.
We managed to match there
the evolution
``before EP'' with
the evolution
``after EP''.
The goal has been realized
via an extreme and
brutal  maximal fine-tuning procedure.
The graduality formula (\ref{grgr})
has been made parameter-dependent, i.e., in our
present notation, $p-$dependent.
Next,
the $p-$supported limiting-confluence
conversion of the sharp inequality signs ``$<$''
into equal signs ``$=$'' in ~(\ref{grgr})
has been
imposed upon {\em all\,} of the separate EP2s.
The EPN degeneracy involved {all} of the states
(the number $N$ of which was chosen even).
The ``gradual''
pattern of Eq.~(\ref{grgr})
has been replaced by its
``confluent'' predecessor (\ref{negrgr}).

In the models of Ref.~\cite{procA}
where $N={\rm dim\,}\,H^{(before\, EP)}(g)
={\rm dim\,}\,H^{(after\, EP)}(g)$,
the construction
implied the complete
degeneracy of the energies,
 \be
 \lim_{g \to g^{(EPN)}}\,E_n(g) = \eta\,,
 \ \ \ \  n = 0, 1, \ldots, N - 1\,.
 \label{conflen}
 \ee
The
phase-transition-mediating
Hamiltonians
acquired, at the matching EP instant, the {\em same},
strongly fine-tuned
canonical
form of a single $N$ by $N$ non-diagonal Jordan-block matrix,
 \be
 \lim_{g \to g^{(EP)}}\,H^{(before/after\, EP)}(g)
 \sim
 J^{({N})}(\eta)=\left [\begin {array}{ccccc}
    \eta&1&0&\ldots&0
 \\{}0&\eta&1&\ddots&\vdots
 \\{}0&0&\eta&\ddots&0
 \\{}\vdots&\ddots&\ddots&\ddots&1
 \\{}0&\ldots&0&0&\eta
 \end {array}\right ] \,.
 \label{hisset}
 \ee
The explicit
construction of a
genuine quantum
energy-level-degeneracy
catastrophe
proved successful and
involved {all of the levels} in the spectrum.

In technical terms,
the feasibility of the construction
reflected the finite-dimensional nature of
Hamiltonians
$H^{(before\, EP)}(g)$ and $H^{(after\, EP)}(g)$.
Although the mechanisms causing the collapse
remained unchanged, the specific simultaneous
EPN-based phase-transition
effect (\ref{conflen}) itself has been
rendered possible.
Hamiltonians
$H^{(before\, EP)}$ and $H^{(after\, EP)}$
were connected and matched
in a strictly
continuous and both-sided
``fundamental-Hamiltonian''
manner.

In the early applications
of the non-Hermitian degeneracies (\ref{hisset}),
many of them
retained their
pragmatic effective-operator open-system physical background
admitting a virtually arbitrary complex $\eta$.
Only in a small minority of the closed-system models
with strictly real spectra
the authors
emphasized their
dynamically complete description as well as their
fundamental-theory
character.

In the latter context of our present exclusive interest,
multiple further new questions emerged.
Some of them will be re-opened and answered in what follows.

\section{Results}

Our present project is aimed at the search for new forms of
manipulation and control
of important qualitative features of quantum dynamics.
Our main
result
can be characterized
as a proposal of an EP-based quantum alternative
to the classical Thom's catastrophe theory.
The essence
of such a classification
concerning the quantum phase transitions
will lie, in its present form, in
the control of
the EP2-related singularities
and, in particular, in the control of their
confluences
and/or restructuralizations.

\subsection{Purpose: Unitary access to EPs in closed quantum systems.}

In the literature,
many authors
(cf., e.g.,
Trefethen and Embree \cite{Trefethen} or
Krej\v{c}i\v{r}\'{\i}k et al \cite{Viola})
studied the PHQM-related
quantum systems far from their
EP singularities. For this reason they
did not need to distinguish
too carefully between the open
(i.e., intrinsically non-unitary) and closed
(i.e., intrinsically unitary) quantum systems.
As a consequence,
several
interpretations of their results
happened to be unclear or even,
involuntarily, misleading.
Typically, whenever they
correctly identified
``unexpectedly wild''
reaction to ``small'' perturbations \cite{Viola},
they did not
emphasize that
such a scenario
is only encountered
in the non-unitary open-quantum-system setup.

The clarification of the apparent puzzle
was published in Refs.~\cite{corridors,admissible}.
For the sake of clarity we
picked up there just
the ``extreme''
matrix~(\ref{hisset}) as an unperturbed operator.
Then, for any perturbed Hamiltonian
 \be
 H^{}(g)
 =
 J^{({N})}(\eta)+V(g)
 \label{herset}
 \ee
we showed that
the
class of perturbations $V(g)={\cal O}(g)$
characterized as ``sufficiently
small'' in the conventional
open-system norm of
the unphysical Hilbert space ${\cal H}_{(auxiliary)}$
has to be re-classified
as unacceptable, always containing
perturbations which prove unbounded
when measured in the correct
closed-system norm of space ${\cal H}_{(physical)}$.

In Ref.~\cite{corridors} these observations were complemented
by the consistent closed-system interpretation of
the perturbed models (\ref{herset})
in ${\cal H}_{(physical)}$.
We demonstrated that
the standard requirement of the smallness of
the norm of
$V(g)$ in ${\cal H}_{(physical)}$
offers a natural picture of reality in
the vicinity of EP. We argued that
in connection with the evolution of models (\ref{herset})
in ${\cal H}_{(physical)}$
one can localize certain
non-empty corridors of unitary access
to the quantum phase transition extremes at
EPs.

A clear
separation between the open- and closed-system theories
must always be kept sufficiently well verbalized.
Partially, what is to be blamed for
the existing misunderstandings
is
the currently widely accepted terminology. Even
our present conventional
usage of the term ``non-Hermitian''
should be taken {\it cum grano salis}, i.e., with
understanding
of its true meaning.
The point is that in the closed-system context
our considerations will never contradict
the conventional formulations of quantum mechanics.
The operators of observables will always be self-adjoint.
The only necessary clarification is that
in the upgraded PHQM SP framework, all of the
computations are
realized
in a manifestly unphysical
Hilbert space ${\cal H}_{(auxiliary)}$ \cite{ali,Geyer}.
The conventional and correct physical
Hilbert space (say, ${\cal H}_{(physical)}$)
remains only available indirectly, via its representation in
${\cal H}_{(auxiliary)}$.

One of the key technical merits of the
PHQM amendment of the theory is that
the latter representation
of ${\cal H}_{(physical)}$ is
mediated by the mere
amendment of the inner product.
The resulting re-arrangements
of the usual SP model-building recipes then really
work with the operators
which are non-Hermitian (in
${\cal H}_{(auxiliary)}$).

\subsection{Tool: Schr\"{o}dinger equations on discrete lattices.}

The main purpose of our present message is to show that
the PHQM enhancement of the flexibility of the SP formalism
leads, near the EP singularities, to some
particularly important consequences.
This will be illustrated by a few not too complicated
benchmark
gain + loss Hamiltonians in which
we will
postulate the existence of two parameters
controlling the strength of the
two separate, independent
gain + loss subcomponents.

Via these toy models we will demonstrate that
one can achieve
several desirable
transmutations of the EPs (i.e.,
of the dynamics in their vicinity)
via the mere fine-tuned
interference between
the remote and central gain-plus-loss interactions.

For introduction let us recall
the ordinary-differential-operator non-Hermitian
square-well model of Ref.~\cite{ptsqw}.
The unitarity (i.e., the reality of the spectrum
of bound states) has only been guaranteed there in a finite
interval of the strength (say, $g$) of the non-Hermiticity.
The reality (i.e., observability) has been lost due to the
EP-related mechanism of the merger of
the ground state with the first excited state,
$\lim_{g \to g^{(EP)}_0}\,[E_0(g)-E_{1}(g)]=0$.

With the further
growth of the non-Hermiticity of $H(g)$
beyond its EP value $g^{(EP)}=g^{(EP)}_0$, further
mergers occurred,
and all of them were followed by the complexifications of the
energies of the higher and higher excited states.
The process
involved, gradually, the whole spectrum,
resulting in the formation of an infinite sequence of
exceptional points $g^{(EP)}$
such that
$\lim_{g \to g^{(EP)}}\,[E_n(g)-E_{n+1}(g)]=0$.

Such a behavior of the EPs appeared to be generic.
Typically, the phenomenologically relevant boundary of
${\cal D}_{(physical)}$
only contained,
in the vast majority of the elementary closed-system models,
a single isolated EP singularity.
A richer, multi-parametric structure of the Hamiltonian
appeared necessary for the realization of any
more interesting scenario.

In order to avoid
the loss of the easy mathematical tractability of the
desirable toy models
we decided to redirect our attention
from the differential
Hamiltonians
$H = -\triangle + V(x)$
to
their difference-operator analogues.
The most straightforward implementation of such an idea is
easy: One simply replaces
the continuous
real line of coordinates $x \in \mathbb{R}$
by an equidistant lattice of grid points
 \be
x_{k+1}=x_{k}+h\,, \ \ \ \ \ \ \ \
  k = 0,1, \ldots, N
   \,.
  \label{RKg}
 \ee
This opens
the possibility of replacement of the conventional
differential Schr\"{o}dinger equation
by its difference-equation analogue
 \be
  -
 \frac{\psi(x_{k+1})-2\,\psi(x_{k})+\psi(x_{k-1})}{h^2}
  +V(x_k)
 \,\psi(x_{k})
 =E\,\psi(x_k)\,, \ \ \ \
  k = 1,2, \ldots, N
 \,.
 \label{diskr}
 \ee
With the equally conventional Dirichlet
asymptotic boundary conditions
$\psi(x_0)=\psi(x_{N+1})=0$
the construction of bound states is then
reduced
to the mere linear algebraic problem
 \be
 \left (
 \begin{array}{ccccc}
 v_1&-1&&&\\
 -1&v_2&-1&&\\
 &-1&\ddots&\ddots&\\
 &&\ddots&v_{N-1}&-1\\
 &&&-1&v_N
 \ea
 \right )\,
 \left (
 \ba
 \psi_1\\
 \psi_2\\
 \vdots\\
 \psi_N
 \ea
 \right )
 =F\,
 \left (
 \ba
 \psi_1\\
 \psi_2\\
 \vdots\\
 \psi_N
 \ea
 \right )\,.
 \label{sematr}
 \ee
In this local-interaction model the Hamiltonian
contains just an $N-$plet of the dynamics-determining diagonal matrix
elements $v_k=h^2V(x_{k})$ yielding the spectrum
of the re-scaled and
shifted bound-state  energies $F_n=h^2E_n-2\ $ with $n=0,1,\ldots, N-1$.

Our interest in Eq.~(\ref{sematr}) was initially inspired by the
popularity of the non-Hermitian Hamiltonians with real spectra
\cite{Carl,Carlbook}. Various non-analytic
square-well realizations of the potentials have been
studied
in this
direction of research \cite{quesne,quesneb,quesnec,quesnek}.
In
these analyses an important
role was played by the
discrete models
as sampled by Eq.~(\ref{sematr}) \cite{matcha,discr3,discr3d,match}.

In an introductory methodical remark let us
pick up $N=6$ and let us
consider Eq.~(\ref{sematr}) with one of the most elementary
constant-interaction
Hamiltonians
 $$
 H^{(6)}(w)=
\left[ \begin {array}{cccccc} -iw&-1&0&0&0&0
\\\noalign{\medskip}-1&-iw&-1&0&0&0\\\noalign{\medskip}0&-1&-iw&-1&0&0
\\\noalign{\medskip}0&0&-
1&iw&-1&0\\\noalign{\medskip}0&0&0&-1&iw&-1\\\noalign{\medskip}0&0&0&0
&-1&iw\end {array} \right]\,.
 $$
The brute-force numerical analysis
reveals that
in spite of the non-Hermiticity of the matrix, its
spectrum is real (i.e., in principle, observable) inside a unique
unitarity-compatible interval of $$w \in {\cal D}_{(physical)}
\approx (-0.322,0.322).$$
Along the whole real line of
parameters
the model supports the existence of as many as
four separate exceptional points,
viz.,
 $$
  \{- 0.54006,\,-0.32215,\,0.32215,\,0.54006\}\,.
 $$
The distance of the outer pair of these EPs
from ${\cal D}_{(physical)}$ is not zero so that they
cannot play any immediate physical
role. Their existence is only considered interesting
in mathematics
(or in the open-system physical setup)
where people are trying to describe,
irrespectively of the condition of
unitarity, the whole
spectrum.

\section{Realization: Models with two free parameters.}

Most of the above-mentioned
studies
confirmed the expectations that at the
sufficiently small non-Hermiticities the spectra of the energy
eigenvalues should be real \cite{Carl,ali,Langer}.
In our present
paper we intend to complement these results by the study of
models in which, via the manipulation of the EPs,
one could control the qualitative dynamics directly.
In a search for such
models one intends to control the
positions of EPs using several independent variable parameters.
In such a project
the main obstacles would be technical
because
besides a few most elementary matrix structures
the brute-force
numerical localization of the EPs is a
badly ill-conditioned problem \cite{nev,Kgt1,neu}.

This being said,
the comparatively transparent and feasible
study of the EPs
can still be based on Schr\"{o}dinger Eq.~(\ref{sematr})
in which
almost all of the
matrix elements
$v_k=h^2V(x_{k})$
of the local interaction term
would be assumed to vanish.
In our present paper we will study, first of all,
the two-parametric family of Hamiltonians
 \be
 H^{(N)}(\varrho,w)=
 \left[ \begin {array}{ccccc|ccccc}
 -{\rm i}{{} \varrho}&-1&0&\ldots&0&0&\ldots&
 &\ldots&0
 \\
 -1&{{} {0}}&-1&\ddots&\vdots&\vdots&\ddots&&&\vdots
 \\
 0 &\ddots&\ddots&\ddots&0&\vdots&&&&
 \\
 \vdots
 &\ddots&-1&{{} {0}}&-1&0&&&\ddots&\vdots
 \\
 0&\ldots&0&-1&-
 {\rm i}\,w&-1&0&\ldots&\ldots&0
 \\
 \hline
 0&\ldots&\ldots&0&-1&{\rm i}\,w&-
 1&0&\ldots&0
 \\
 \vdots&\ddots&&&0&-1&{{} {0}}&-1&\ddots&\vdots
 \\
 &&&&\vdots&0&\ddots&\ddots&\ddots&0
 \\
 \vdots
 &&&\ddots&\vdots&\vdots&\ddots&-1&{{} {0}}&-1
 \\
 {}0&\ldots&&\ldots&0&0&\ldots&0&-1&
 {\rm i}{{} \varrho}
 \end {array} \right]\,.
 \label{latti}
 \ee
in which $N=2K$ is even
and in which the central and remote parts of the interaction
(with the respective strengths $w$ and $\varrho$) are well separated.

\subsection{The confluence of EPs controlled by the
fine-tuning of the remote gain-and-loss
interaction.}

Technically, the separation of the influence of $w$ and $\varrho$
can simply be strengthened, whenever needed, by the choice of
a sufficiently large matrix dimension $N$.
At the same time,
the potentially
adverse aspect of the growth of $N$
(making the secular equation less easily tractable)
can very easily be suppressed
using the dedicated
$N-$independent matching method of Ref.~\cite{match}.
Using this method one can always try to test whether
the bound-state spectrum of the
closed-system
toy-model Hamiltonian (\ref{latti}) is real.

Usually, the answer becomes affirmative for the parameters
lying inside a two-dimensional unitarity-compatible
(and, say, real) domain ${\cal D}_{(physical)}$.
Within the framework of our present project
we will only be interested in the situations in which
one of the parameters is fixed while the other one approaches
the boundary $\partial {\cal D}_{(physical)}$
of the energy-reality domain.
What one then {\it a priori\,} expects is that for the
different choices of the fixed parameter
the mergers of the energies encountered
at the EP boundary
might be of different types.

\begin{figure}[h]                    
\begin{center}                         
\epsfig{file=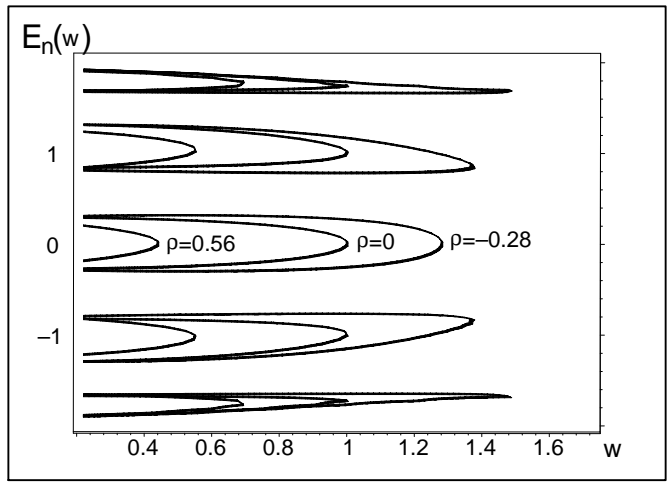,angle=270,width=0.5\textwidth}
\end{center}    
\caption{The $N-$plets of the real energy eigenvalues
$E_n({\varrho},w)$ of Hamiltonian (\ref{latti}) with $N=10$ and
$n=0,1,\ldots,N-1$ at constant ${\varrho}=-0.28 $ (the rightmost
curves), at constant ${\varrho}=0 $ (the curves in the middle) and
at constant ${\varrho}=0.56 $ (the leftmost curves).
 \label{ouqa}}
\end{figure}





In the first test of the hypothesis let us
choose $N=10$. Once we fix
the remote-gain-and-loss parameter ${\varrho}$
we may study the
spectra of energies $E_n({\varrho},w)$, $n=0,1,\ldots,N-1$
as functions of the remaining variable parameter $w$.
Numerically we evaluated several characteristic samples
of such a type.
In Fig.~\ref{ouqa} we displayed three of them. The energies are
shown there as functions of $w$,
calculated at the three different values of
the remote-non-Hermiticity parameter ${\varrho}$,
viz., at
${\varrho}=-0.28$ (the rightmost curves), of
${\varrho}=0$ (the middle-positioned curves), and of
${\varrho}=0.56$ (the leftmost curves).
After inspection of this picture it is possible to formulate
the following observation.

\begin{conj}
\label{kojed}
At the
sufficiently small real values of the remote-interaction
parameter ${\varrho}$
the separate EP2 coordinates $w=w^{(EP2)}_j({\varrho})$
of the pairwise mergers
of the neighboring energies $E_{2j}^{(10)}({\varrho},w)$
and $E_{2j+1}^{(10)}({\varrho},w)$
are all strictly decreasing functions of ${\varrho}$ such that
 \be
 w^{(EP2)}_2({\varrho})\leq
 w^{(EP2)}_1({\varrho})=
 w^{(EP2)}_3({\varrho})\leq
 w^{(EP2)}_0({\varrho})=
 w^{(EP2)}_4({\varrho})\,.
 \ee
The inequalities become sharp at ${\varrho} \neq 0$.
At larger $w>w^{(EP2)}_j({\varrho})$
the respective pairs of energies become complex so that
the local boundary of ${\cal D}_{(physical)}$
becomes determined by
function $w^{(EP2)}_2({\varrho})$.
\end{conj}

Beyond such a purely numerically supported hypothesis
(which could probably be generalized to hold at any
matrix dimension $N=2K$),
the inspection of Fig.~\ref{ouqa} also inspired the
formulation of the following exact result valid at all $K$s.


\begin{prop}
\label{propa}
During the passage of
the remote coupling $\,\varrho\,$ through the origin
at $\varrho=0$, Hamiltonian
(\ref{latti}) encounters the instantaneous confluence
of all of the separate exceptional points of order two,
 \be
 w^{(EP2)}_j(0)=1\,,
 \ \ \ j = 0,1,\ldots,N-1\,.
 \ee
The canonical form of
the $w \to 1$ limit of matrix (\ref{latti})
then
acquires the $N$ by $N$ matrix form
 \be
 H^{(2K)}(0,1)\ \sim \
 \bigoplus_{j=1}^K\,
 J^{(2)}(\eta_j)
 \label{mafo}
 \ee
of a direct sum of $K$
two-dimensional Jordan blocks as defined in Eq.~(\ref{hisset}).
\end{prop}

 \noindent
In the limit of ${\varrho} \to 0$,
we witness the complete
degeneracy {\it alias\,} confluence
of all of the separate EP2s.
The rigorous proof will be given
in section \ref{sese} below. As a byproduct of this proof, also
the values of the limiting energies $\eta_j$ will be shown obtainable
in closed form.

\subsection{The confluence of EPs caused by the fine-tuning of
the central gain-and-loss
interaction.}

Our above-outlined projection
of the motion of the EP boundaries of the
two-dimensional
domain ${\cal D}_{(physical)}$
can be complemented
by the perpendicular sections
in which the value of $w$ is fixed.
We may expect that
at the boundary $\partial {\cal D}_{(physical)}$
the energies will
merge in a way
reflecting the characteristics of the
underlying EPs.

%
%

\begin{figure}[h]                    
\begin{center}                         
\epsfig{file=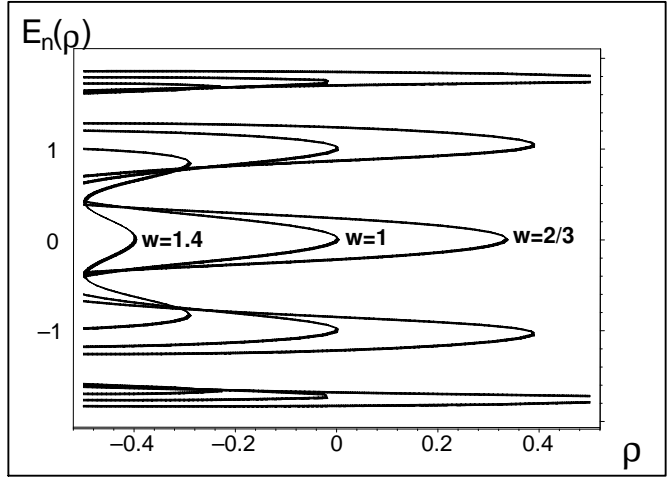,angle=270,width=0.5\textwidth}
\end{center}    
\caption{The $N-$plets of the real energy eigenvalues
$E_n({\varrho})$ of Hamiltonian (\ref{latti}) with $N=10$ and
$n=0,1,\ldots,N-1$ at constant ${w}=2/3 $ (the rightmost curves), at
constant ${w}=1 $ (the curves in the middle) and at constant
${w}=1.4 $ (the leftmost curves).
 \label{ouqas}}
\end{figure}


%




The resulting scenario sampled in
Fig.~\ref{ouqas}
is qualitatively not too different from
its predecessor of
Fig.~\ref{ouqa}.
At every fixed value of $w$,
the single central
EP2 energy-merger ${\varrho}^{(EP)}_2$ is smaller than its first
off-central doubly-degenerate
partner ${\varrho}^{(EP)}_1={\varrho}^{(EP)}_3$ which is, in its turn,
smaller than the second,
most
off-central partner doublet ${\varrho}^{(EP)}_0={\varrho}^{(EP)}_4$.
With the growth of $w$
sampled, in the picture, by the choice of $w=2/3$, $w=1$ and $w=1.4$,
all of the energy mergers
move leftwards.

\begin{conj}
\label{nejed}
In a vicinity of $w=1$
the separate EP2 coordinates $\varrho=\varrho^{(EP2)}_j({w})$
of the pairwise mergers
of the neighboring energies $E_{2j}^{(10)}({\varrho},w)$
and $E_{2j+1}^{(10)}({\varrho},w)$
are all strictly decreasing functions of ${w}$ such that
 \be
 \varrho^{(EP2)}_2({w})\leq
 \varrho^{(EP2)}_1({w})=
 \varrho^{(EP2)}_3({w})\leq
 \varrho^{(EP2)}_0({w})=
 \varrho^{(EP2)}_4({w})\,.
 \ee
The inequalities are certainly sharp at ${w} \neq 1$.
At the small and positive difference
$\varrho-\varrho^{(EP2)}_j({w})>0$
the respective pairs of energies cease to be real so that
the local boundary of ${\cal D}_{(physical)}$ is
prescribed by
function $\varrho^{(EP2)}_2({w})$.
\end{conj}

In the limit of $w \to 1$ all of the EP2 singularities
may be guessed to coincide
forming a degenerate quintuplet EP10.
Such a possibility
advised by
the inspection of Fig.~\ref{ouqas} inspired the
following proposition in which the value of
the even matrix dimension $N=2K$ can be arbitrary.


\begin{prop}
\label{propb}
During the passage of
the central coupling $\,w\,$ through the value of
$w=w^{(EP)}=1$, Hamiltonian
(\ref{latti}) encounters the instantaneous confluence
of all of the separate exceptional points of order two,
 \be
 \varrho^{(EP2)}_j(1)=0\,,
 \ \ \ j = 0,1,\ldots,N-1\,.
 \ee
The canonical form of
the $\varrho\to 0$ limit of matrix (\ref{latti})
acquires the same $N$ by $N$ matrix form (\ref{mafo})
as in Proposition \ref{propa}
\end{prop}
\begin{proof}
In the light of Proposition \ref{propa} the proof is elementary
because our Hamiltonian matrix $H^{(2K)}(\varrho,w)$ can be
transformed into matrix $H^{(2K)}(w,\varrho)$
using an elementary block-diagonal
unitary-transformation matrix ${\cal U}^{(2K)}$
defined as a direct sum of two $K$ by $K$
antidiagonal unit matrices ${\cal I}^{(K)}$ with Kronecker-delta
elements ${\cal I}^{(K)}_{i,j} ={\delta}_{i,K+1-i}$, $i=1,2,\ldots,K$.
\end{proof}

\section{Exact solutions at $\varrho=0$.\label{sese}}

Our illustrative Hamiltonian (\ref{latti})
is the special case of a
broader class of models of Eq.~(\ref{sematr}). Some of
them might be analytically solvable
by the matching method of Ref.~\cite{match}.
What would be required is a special
choice of the matrix elements $v_j$
of the interaction.
For the sake of simplicity we decided to
consider
just the very special model
(\ref{latti}),
with the study of its possible
generalizations left to the readers.

\subsection{Constructive proof of proposition \ref{propa}.}

The independent variability of the two real parameters
$\varrho$ and
$w$ in  (\ref{latti})
proved sufficient for our present
illustration purposes.
Now we only have to prove Proposition \ref{propa}
in which
our $\varrho=0$ toy-model Hamiltonian
has even $N=2J+2$ and mere
two nonzero
values of $v_j$,
 \be
 H^{(2J+2)}(w)=
  \left ( \begin {array}{cccc|cc|cccc}
  0&-1&0&\ldots&0&0&0&\ldots&\ldots&0
 \\{}
 -1&0&-1&\ddots&\vdots&\vdots&\vdots&&&\vdots
 \\{}
 0&-1&\ddots&\ddots &0&\vdots&&&&
 \\{}
 \vdots&\ddots&\ddots&0&-1&0&\vdots&&&\vdots
 \\
 \hline
 {}0&\ldots&0&-1&-{\rm i}w&-1&0&\ldots&\ldots&0
 \\{}0
 &\ldots&\ldots&0&-1&{\rm i}w&-1&0&\ldots&0
 \\
 \hline
 {}\vdots &&&\vdots&0&-1&0&-1&\ddots&\vdots
 \\{}
 &&&&\vdots&0&-1&\ddots&\ddots&0
 \\{}
 \vdots&&&\vdots&\vdots&\vdots&\ddots&\ddots&0&-1
 \\{}
 0&\ldots&\ldots&0& 0&0&\ldots&0&-1&0
 \end {array} \right )
 \label{mytoy}
 \ee
The related Schr\"{o}dinger equation is
tractable
by the standard numerical diagonalization techniques.
The task becomes simplified when one
introduces,
in the spirit of Refs.~\cite{BB,BG},
the requirement
${\cal PT}\,H^{(2J+2)}(w)=H^{(2J+2)}(w)\,{\cal PT}$ of
${\cal PT}-$symmetry defined
in terms of the
antidiagonal-unit matrix ${\cal P}$ [changing the parity and causing
the left-right inversion of the spatial lattice (\ref{RKg})] and
of the
antilinear complex-conjugation operator ${\cal T}$ (which simulates
the time reversal in Schr\"{o}dinger equation).
Whenever
our real parameter $w$ is such that
the spectrum remains real and non-degenerate,
Schr\"{o}dinger equation (\ref{sematr}) will then yield,
exclusively, just the ${\cal
PT}-$symmetric eigenstates,
 \be
  {\cal PT}\,
 \left (
 \ba
 \psi_1\\
 \psi_2\\
 \vdots\\
 \psi_N
 \ea
 \right )  \sim
 \left (
 \ba
 \psi_1\\
 \psi_2\\
 \vdots\\
 \psi_N
 \ea
 \right )
%
 \,.
 \label{petes}
 \ee
We may set $\psi_N=\psi_1^*$, $\psi_{N-1}=\psi_2^*$ (etc), we may
abbreviate $h^2E=-2x=-2\,\cos \theta$, and
we may recall the definition of
the Tshebyshev plynomials of the second kind,
 \be
 U_k(\cos \theta)=\frac{\sin (k+1)\theta}{\sin \theta}\,,
 \ \ \ \ k = 0, 1, \ldots\,.
 \label{forka}
 \ee
The recurrences satisfied by these
polynomials \cite{cebysevs,cebysevsk} enable us to guess the
ansatz
 \ben
 \psi_{k+1}=(\alpha+{\rm i}\,\beta)\,U_k(x)\,,
 \ \ \ \ k = 0, 1, \ldots, J\,
 \een
containing
just a pair of unspecified real parameters $\alpha$ and $\beta$.
Its use
converts the first $J$ lines of relations (\ref{sematr})
into identities. Recalling the  ${\cal PT}-$symmetry of the
model we may also write down the rest of the components of the
eigenvector in closed form reflecting the validity
of the last $J$ lines of relations
(\ref{sematr}),
 \ben
 \psi_{N-k}=(\alpha-{\rm i}\,\beta)\,U_k(x)\,,
 \ \ \ \ k = 0, 1, \ldots, J\,.
 \een
What remains to be satisfied are the two middle lines of
Schr\"{o}dinger equation~(\ref{sematr}),
 \ben
 -(\alpha+{\rm i}\,\beta)\,U_{J-1}(x)
 +[(2x-{\rm i}w)\,(\alpha+{\rm i}\,\beta)-(\alpha-{\rm i}\,\beta)]\,
 U_{J}(x)=0\,,
 \een
 \ben
 -(\alpha-{\rm i}\,\beta)\,U_{J-1}(x)
 +[(2x+{\rm i}w)\,(\alpha-{\rm i}\,\beta)-(\alpha+{\rm i}\,\beta)]\,
 U_{J}(x)=0\,.
 \een
The separation of the
real and imaginary components
yields
 \ben
 -\alpha\,U_{J-1}(x)
 +(2x\alpha +w\,\beta)-\alpha)\,
 U_{J}(x)=0\,
 \een
and
 \ben
 -\beta\,U_{J-1}(x)
 +(2x\beta-w\,\alpha+\beta)\,
 U_{J}(x)=0\,.
 \een
After a premultiplication by suitable constants,
the sum of the latter two relations yields
 \be
 -2\,\alpha\,\beta\,U_{J-1}(x)+[4\,\alpha\,\beta\,x
 +(\beta^2-\alpha^2)\,w]\,U_{J}(x)=0
 \label{former}
 \ee
while their difference only leads to elementary relation
 \be
 (\alpha^2+\beta^2)\,w=2\,\alpha\,\beta\,.
 \ee
This enables us to reparametrize $\alpha=\alpha(\tau)=\cos \tau$ and
$\beta=\beta(\tau)=\sin \tau$ and to deduce that $w=w(\tau)=\sin
2\tau$.

One can treat the auxiliary angle $\tau$ as an
alternative dynamical-input information about the strength of the
non-Hermiticity.
We are now only left with the secular equation (\ref{former}), i.e.,
 \be
 2\,\alpha\,\beta\,U_{J+1}(x)
 +(\beta^2-\alpha^2)\,w\,\,U_{J}(x)=0\,.
 \ee
The insertion of $w$ reduces it to the relation
 \be
 U_{J+1}(x)=
 [\alpha^2(\tau)-\beta^2(\tau)]\,U_{J}(x)\,.
 \label{secueq}
 \ee
This is our ultimate implicit definition of the spectrum of the
energies $h^2E=-2x$ at arbitrary matrix dimension $N=2J+2$.

At the ${\cal PT}-$symmetry-breakdown
boundaries with $w=w^{(EP)}=\pm 1$ or $\tau=\tau^{(EP)}=\pm \pi/4$,
we have $\alpha^2(\tau^{(EP)})=\beta^2(\tau^{(EP)})$ so that
the EPN-related energy values coincide with the
roots of a single polynomial,
 \be
 U_{J+1}\left (x^{(EP)}\right )=0\,.
 \label{facN}
 \ee
These roots
can be given an elementary form given by
formula (\ref{forka}).

The availability of such an explicit
parameter-dependence of the spectrum
in the EPN limit
can be extended to cover also, in an approximative form, a
small vicinity of the singularity.
In this vicinity the difference
$\alpha^2(\tau)-\beta^2(\tau)$
entering Eq.~(\ref{secueq}) will be a small number. The
well known intertwining property of the roots of
the polynomials $U_{J+1}(x)$ and $U_{J}(x)$
will then immediately imply
the correct qualitative understanding of the branching of the levels
at $|w| \lessapprox  1$ as sampled, at $N=10$, in Fig.~\ref{lo6ja3}.

%

\begin{figure}[h]                    
\begin{center}                         
\epsfig{file=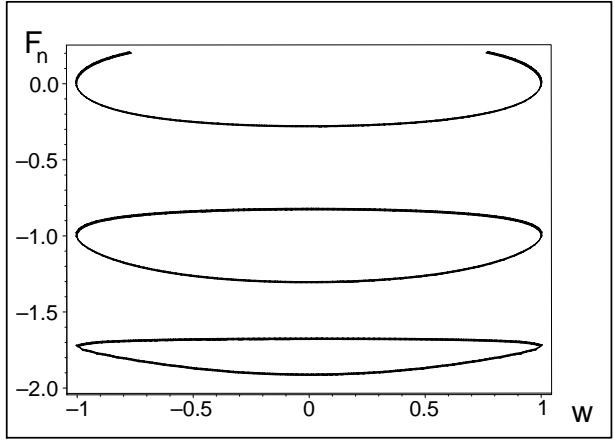,angle=270,width=0.460\textwidth}
\end{center}    
\vspace{2mm} \caption{The
left-right symmetry of the
low-lying spectrum of model (\ref{mytoy})
at $N=10$. The spectral locus is
also symmetric
with respect to the $F=0$
axis, so we did not need to display the
upper, high-excitation half of the spectrum.
 \label{lo6ja3}
 }
\end{figure}

%
%
%

\subsection{Spectral curves at $N=10$.}

The interval  $w \in (-1,1)$
of the unitarity-compatible
physical parameters
is
the same for {\em
all\,} of the bound states
at an {\em
arbitrary\,} even matrix dimension $N=2J+2$.
Whenever the value of $w$ leaves this interval, {\em
all\,} of the energies cease to be real so that abruptly,
the {\em whole\,} spectrum
becomes unobservable.
At all of the
lattice-size-determining integers $J$, even
the parameter-dependence of the
numerically evaluated spectra remains
qualitatively the same,
characterized by the specific {\em pairwise\,} degeneracies of {\em
all\,} of the energy levels at $w=1$ and $w=-1$.

At $N=10$ the  model still remains non-numerical.
Its secular polynomial $P^{(10)}({F},{w})$ is a polynomial
of the fifth degree in the energy-representing variable $F^2$,
 \be
  {{\it {F}}}^{10}+ \left( -9+{{\it {w}}}^{2} \right) {{\it {F}}}^{8}+
 \left( 28-6\,{{\it {w}}}^{2} \right) {{\it {F}}}^{6}+ \left( -35+11\,{
{\it {w}}}^{2} \right) {{\it {F}}}^{4}+ \left( 15-6\,{{\it {w}}}^{2}
 \right) {{\it {F}}}^{2}-1+{{\it {w}}}^{2}\,
 \label{desit}
 \ee
but
this does not imply that the search for its roots
is complicated.
Their brute-force numerical localization is not necessary.
It is sufficient to notice that
the secular polynomial
is just a linear function of the square of the coupling
constant $w^2$. This implies that the spectrum
can be constructed,
in the implicit-function form, non-numerically.

The shape
and symmetry of the spectrum are sampled in Fig.~\ref{lo6ja3}.
The picture just
reconfirms the existence of the strictly two Kato's
exceptional points $w =w^{(EP)}_\pm = \pm 1$.
One can visualize the
coupling $w$
as an
{\em elementary\,}  function
$$
 w=w_\pm(F) = \pm\, [\left( {{F}}^{2}-1 \right)^2  - {{F}}^{2} ]^{-1}\,
{\sqrt {1-{{F}}^{10}+9\,{{F}}^{8}-28\,{{F}}^{6}+35\,{{F}}^{4}-15\,{{F}}^{
2}}}\,
$$
of the energy.
At both of the EPN extremes with $|w|=1$
the $N=10$ secular polynomial
can be factorized,
 \be
 P^{(10)}({F},\pm 1)={F}^2\,({F}^2-1)^2\,({F}^2-3)^2\,.
 \label{fac10}
 \ee
The function $w=w(F)$ can be also Taylor-expanded.
Near $F=0$ this yields the symmetric
and ``deeper-than-quadratic'' well,
 $$
 w({F}) \approx -1+{\frac {9}{2}}{{F}}^{2}
 +{\frac {201}{8}}{{F}}^{4}+O \left( {{F}}^{5}
 \right) \,.
 $$
Similarly,
off the origin we get, in agreement with the picture,
the two narrower and asymmetric wells which are
steeper than the one near the origin. Thus, we get
 $$
 w({F}) \approx -1+8\, \left( {F}-1 \right) ^{2}-24\, \left( {F}-1 \right) ^{3}+122\,
 \left( {F}-1 \right) ^{4}+O \left(  \left( {F}-1 \right) ^{5} \right) \,,
 $$
etc. Finally, the outer wells have just the more pronounced shapes of the same form,
with
 $$
 w({F}) \approx
 -1 +
 72\, \left( {F}- \sqrt{3} \right) ^{2}- 648\,\sqrt {3}\,
 \left( {F}- \sqrt{3} \right) ^{3}+ 16782\, \left( {F}-
 \sqrt{3} \right) ^{4}+O \left(  \left( {F}- \sqrt{3} \right) ^{5}
 \right) \,
 $$
etc. All of these observations
gave birth to their generalizations valid in any analogous
EP-supporting $N-$level quantum system with arbitrary finite $N<\infty$.


\section{Discussion}

From the purely methodical point of view
the choice of the discrete local-interaction
model of Eq.~(\ref{sematr})
has its weaknesses. Firstly,
its variable parameters only
lie on the main diagonal.
This lowers the flexibility of dynamics leading, typically,
just to the
EP2 energy mergers.
Secondly,
additional
antilinear symmetries [sampled here by
${\cal PT}-$symmetry of
Eq.~(\ref{petes})]
had to be imposed
in order to guarantee the reality of the spectrum.
Thirdly,
the well known numerically ill-conditioned nature
of the study of the limiting transition $g \to g^{(EP)}$
often forces us to
use certain truly sophisticated construction methods
in a way sampled, say, in Ref.~\cite{five}.

For all of these reasons it will make sense to
turn attention to the
more general matrix
models in which
the practical calculations remain feasible
but in which
it should still be possible
to enhance
the flexibility of the picture of the EP-related dynamics.
Let us now mention a few hints for the future
projects oriented in this direction.

\subsection{Parallels between harmonic oscillator
and our $N<\infty$ models.}

The turn of attention to the
more general classes of models
might
open
new ways towards an immediate
further development of the theory itself.
In order to be more specific
let us
recall, once more, the
harmonic oscillator results as sampled in
Fig.~\ref{pictho} and in Eq.~(\ref{dekq}) above.
In place of the
canonical Jordan-block limit
of Eq.~(\ref{hisset})\, one obtains, for them,
an alternative,
infinite-dimensional but partitioned Jordan-block limit
 \be
 \lim_{g \to g^{(EP)}}\,H^{(HO)}(g)
 \sim
 \left(
 \begin{array}{cc|cc|cc}
 2&1&0&0&0&\ldots\\
 0&2&0&0&0&\ldots\\
 \hline
 0&0&6&1&0&\ldots\\
 0&0&0&6&0&\ldots\\
 \hline
 0&0&0&0&10&\ldots\\
 \vdots&\vdots&\vdots&\vdots&\ddots&\ddots
 \ea
 \right ) \,
 \label{onosset}
 \ee
of the form of
Eq.~(\ref{mafo}) with infinite
sequence of the
energy
mergers available in closed form, $\eta_j=4j-2$ \cite{scirep}.

The EP singularities of
our present $N < \infty$ models (\ref{latti})
lead to
an analogous canonical-representation limit
with the known values of $\eta_j$. In particular,
for our $N=10$ model (\ref{mytoy}) characterized by the
secular polynomial of Eq.~(\ref{desit}) and by
the canonical form (\ref{mafo}) of the EP10 limit
with $K=5$, it would be easy to construct the so called
transition matrices $Q^{{(10)}}_{}$ and to evaluate
the
canonical-representation form
of the Hamiltonian,
 \be
 H^{{(canonical)}}_{}(w)\,
 =[Q^{{(10)}}_{}]^{-1} \,H^{{(10)}}(w)\,.
  Q^{{(10)}}_{}
 \,
 \label{Crealt}
 \ee
In the EP limit we would get
 \be
 \lim_{w \to w^{(EP)}}\,H^{(10)}(w)
 \sim
 H^{{(canonical)}}_{}(1)=
 \left(
 \begin{array}{cc|cc|cc|cc|cc}
 -\sqrt{3}&1&&&&&&&&\\
 0&-\sqrt{3}&&&&&&&&\\
 \hline
 &&-1&1&&&&&&\\
 &&0&-1&&&&&&\\
 \hline
 &&&&0&1&&&&\\
 &&&&0&0&&&&\\
 \hline
 &&&&&&1&1&&\\
 &&&&&&0&1&&\\
 \hline
 &&&&&&&&\sqrt{3}&1\\
 &&&&&&&&0&\sqrt{3}\\
 \ea
 \right ) \,.
 \label{eosset}
 \ee
Such a canonical-Hamiltonian matrix
is block-diagonal. At a fixed
algebraic multiplicity of EPN (i.e., at $N=10$ in this case)
the number $K$ of its
independent eigenvectors (called the geometric multiplicity of EPN)
is maximal (here, we have $K=5$).


\subsection{Asymmetric real-matrix models.}

One of the
next-step model-building strategies
could be inspired by the less explored
non-numerical constructions of Refs.~\cite{tridiagonal,trib}.
The
necessary simplification of the technicalities
has been achieved there by the reduction of the class of
the eligible
Hamiltonians to the mere tridiagonal real
and real-parameter-dependent asymmetric matrices
admitting off-diagonal interaction terms.
In contrast to our preceding models,
the
weakly non-local interactions of such a type
proved useful, e.g., in
the pseudo-Hermitian models of scattering
\cite{Jones,scattb,scattc,scattj,scattf,Kuzh}.

For an illustration of their specific merits let us
recall now the six-by-six-dimensional special case of
the $N$ by $N$ matrices of Ref.~\cite{tridiagonal},
and let us
complement it by an ${\cal O}(g)$
perturbation. This leads to
one of the most
user-friendly real-matrix two-parametric Hamiltonians, viz.,
 \be
  H^{(toy)}(g,\lambda)=\left[ \begin {array}{cccccc}
 -5+g &\sqrt {5+5\,{\it \lambda}}&0&0&0&0
 \\\noalign{\medskip}-\sqrt {5+5\,{\it \lambda}}&-3&2\,\sqrt {2+2\,{\it \lambda
 }}&0&0&0
 \\\noalign{\medskip}0&-2\,\sqrt {2+2\,{\it \lambda}}&-1&3\,\sqrt {
 1+{\it \lambda}}&0&0
\\\noalign{\medskip}0&0&-3\,\sqrt {1+{\it \lambda}}&1&2\,
  \sqrt {2+2\,{\it \lambda}}&0
 \\\noalign{\medskip}0&0&0&-2\,\sqrt {2+2\,{
 \it \lambda}}&3&\sqrt {5+5\,{\it \lambda}}
 \\\noalign{\medskip}0&0&0&0&-\sqrt {
 5+5\,{\it \lambda}}&5-g\end {array} \right]\,.
 \label{ha6}
 \ee
The real, non-degenerate and
equidistant spectrum is obtained
in a $g-$dependent unitarity-compatible
interval ${\cal D}^{(toy)}(g)$ of parameters $\lambda$.
The simplest
proof becomes available
in the unperturbed case with $g=0$. One merely has
to recall the closed formulae of Ref.~\cite{tridiagonal} yielding
the admissibility interval of $\lambda \in (-\infty,0)$
or, in the real-matrix case,
the narrower range of $\lambda \in (-1,0)$.

\subsection{The parameter-controlled
change of the geometric multiplicity.}

 \noindent
In model (\ref{ha6})
with small $g>0$ we may
omit, as trivial, the
half-line of parameters $\lambda < -1$
at which the matrix becomes complex but Hermitian.
We are left with the variability of the single unitarity-supporting
parameter $\lambda \in
(-1,\mu(g))={\cal
D}^{(toy)}(g)$ with $\mu(g) \leq 0$, and we notice that the
parameter-dependence of the
spectrum is entirely different from
our preceding models.
At the left boundary
$\lambda=-1$ the matrix $H^{(toy)}(g,-1)$ becomes diagonal, i.e.,
Hermitian and
tractable as a truncated conventional harmonic oscillator with
equidistant spectrum.
In contrast, the spectral pattern
is very different at the right boundary of ${\cal
D}^{(toy)}$: See
Fig.~\ref{glouqa}
for illustration.

\begin{figure}[h]                    
\begin{center}                         
\epsfig{file=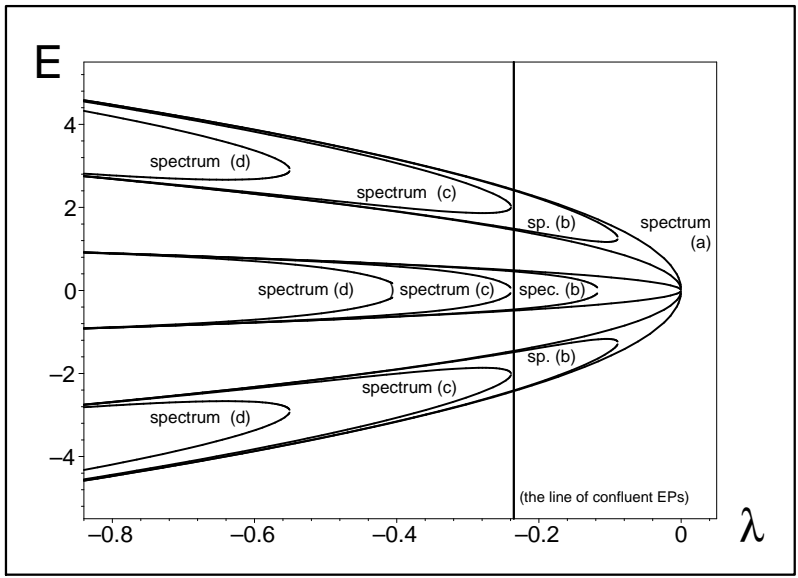,angle=270,width=0.6\textwidth}
\end{center}    
\caption{The $\lambda-$dependent spectra of Hamiltonian (\ref{ha6})
at $g=0$ [spectrum (a)],  $g=1/500$ [spectrum (b)], $g=1/40$
[spectrum (c)] and $g=1/5$ [spectrum (d)].
 \label{glouqa}}
\end{figure}

The comparatively elementary nature of the model
facilitates a
detailed interpretation of the
shapes of the spectra. At the smallest $g$s the value
of the upper bound $\mu(g)$ is
determined by the central energy merger representing an EP of order
two (EP2). We may set
$\mu(g)= \lambda^{(EP)}_1$.
In the opposite extreme of a
large shift $g$ one gets
another, different behavior and formula for $\mu(g)= \lambda^{(EP)}_0=
\lambda^{(EP)}_2$.
The change of
the pattern clearly reflects the fragility of the off-cental
states exhibiting, at larger $g$s,
the
confluence of their EP2 singularities
(in the notation of Ref.~\cite{without} one would write EP=EP2$\oplus$EP2).

The main {\em qualitative\,} difference from the dynamics of
the ``local''
models (\ref{latti}) can be now formulated as the following observation.

\begin{conj}
In model (\ref{ha6})
the separate generic
$g-$dependent
EP2 coordinates $\lambda=\lambda^{(EP2)}_j({g})$
of the pairwise mergers
of the neighboring energies $E_{2j}^{(10)}({\lambda},g)$
and $E_{2j+1}^{(10)}({\varrho},g)$
with $j=0, 1,2$
are all strictly decreasing functions of ${g} \in (0,1/5)$ such that
 \be
 \lambda^{(EP2)}_1({g})<
 \lambda^{(EP2)}_0({g})=
 \lambda^{(EP2)}_2({g})<0\,
 \ee
near the origin (i.e., for $g\ll 1/40$), and such that
 \be
 0>
 \lambda^{(EP2)}_1({g})>
 \lambda^{(EP2)}_0({g})=
 \lambda^{(EP2)}_2({g})\,
 \ee
at the larger $g \gg 1/40$.
\end{conj}

 \noindent
From this bracketing feature one can immediately deduce
the following obvious result which we will give here
without a formal proof.

\begin{prop}
In model (\ref{ha6}) there exists
an EP6=EP2$\oplus$EP2$\oplus$EP2 singularity
$g^{(EP6)} \in (0,1/5)$ with the geometric multiplicity $K=3$.
\end{prop}

 \noindent
The approximate numerical estimate of $g^{(EP6)} \approx 1/40$
has been used to display the
corresponding $\lambda-$dependence of the
energies in Fig.~\ref{glouqa} [marked as ``spectrum (c)''].
Obviously, the local boundary of the physical interval
${\cal D}^{(toy)}$
is given by the
function $\lambda^{(EP2)}_1({g})$ at $g<g^{(EP6)}$, and by the
function $\lambda^{(EP2)}_0({g})=\lambda^{(EP2)}_2({g})$
at $g>g^{(EP6)}$.

We may conclude that our toy model (\ref{ha6})
admits a smooth transition from the
dynamical
regime with the
minimal geometric multiplicity
$K=1$ of the EP6 at $g=g^{(EP6)}_{(K=1)} =0$
to its maximal-geometric-multiplicity alternative
with $K=3$ at $g=g^{(EP6)}_{(K=3)} \approx 1/40$.
In the former case all of
the bound-state energies
converge, in a way
prescribed by Eq.~(\ref{conflen}), to the single EP
value $\eta=E^{(EPN)}$ with $N=6$
at
$\lambda^{(EP)}=\lambda^{(EP)}(g)=\lambda^{(EP)}(0)$,
 \be
 \lim_{{\lambda}\to {\lambda}^{(EP)}}
 E_{j}({\lambda})=E^{(EP6)}\,,\ \ \ \ j=0,1,\ldots,5\,.
 \label{dedekq}
 \ee
In \cite{tridiagonal},
via solvable tridiagonal real-matrix models we managed to simulate
such a
minimal geometric multiplicity behavior of the energies
for an arbitrary
preselected finite Hilbert-space dimension $N < \infty$.
Using a brute-force numerical search
such a type of construction
with minimal $K=1$ remains feasible even in the models which are
realistic \cite{bhgen}.

In the opposite extreme of the dynamical scenario near
EPN=EP2$\oplus$EP2$\oplus \ldots \oplus$EP2 with
the even algebraic EP multiplicity $N=2K$
(such that $K$
now
represents the maximal geometric multiplicity)
the energy degeneracy
is ``maximally incomplete'', having
proceeded merely pairwise,
 \be
 \lim_{{\lambda}\to {\lambda}^{(EP)}}
 E_{n_j}({\lambda})=\eta_j\,,\ \ \ \ j=1,2,\ldots,K\,,
 \ \ \ \
 n_1=0,1\,,\ \
 n_2=2,3\,,\ \ \ldots\,\ \
 n_K=2K-2,2K-1\,.
 \label{ekq}
 \ee
For our model (\ref{ha6})
we just have to insert $K=3$ and
specify
$\lambda^{(EP)}
=\lambda^{(EP)}\left(g^{(EP6)}_{(K=3)}\right)$.

Along similar lines one can
simulate the genuine quantum phase
transition phenomena with an optional geometric multiplicity $K$.
The first applications of such an approach may already be found
in the elementary methodical toy models~\cite{pre},
with the next stage of developments to be aimed at the topical
realistic applications
of the theory,
say, in the descriptions of the
mechanism of the Bose-Einstein condensation
using
the
multi-bosonic pseudo-Hermitian
Bose-Hubbard Hamiltonians \cite{without,preUwe,zaUwe,zaUweb}.

Multiple related mathematical questions remain open.
Nevertheless,
using the standard
Kato's
terminology \cite{Kato},
we certainly will have to distinguish,
at a fixed algebraic EPN multiplicity $N$,
between the occurrence of
a minimal geometric multiplicity $K=1$
[leading to the canonical-representation limit of
Eq.~(\ref{hisset})], of
a maximal geometric multiplicity $K=N/2$
[yielding the alternative canonical-representation limit of
Eq.~(\ref{mafo})],
and of all of the other possibilities in between
these two extremes. This leads us to our final
methodical conclusion.

\begin{prop}
Any given EPN-supporting quantum closed-system Hamiltonian
may be characterized, in its EPN limit,  by its canonical
$N$ by $N$ matrix form
$H^{{(canonical)}}_{}$
with the most general direct-sum
{\it alias\,} block-diagonal-matrix structure
 \be
 H^{{(canonical)}}_{}=
 \bigoplus_{j=1}^K\,
 J^{(N_j)}(\eta_j)\,,
 \ \ \ \ N_1+N_2+\ldots N_K=N\,
 \label{afora}
 \ee
containing nontrivial partitions $N_j \geq 2$.
\end{prop}

 \noindent
The latter operator EPN limit is fully
characterized by the partitioning of $N$
(check some of its number-theory aspects in \cite{Acc})
and by the $K-$plet of the EPN energies $\eta_j$.
Thus, every classification of phase transitions should refer to
the pair of the multiplicites $N$ (algebraic) and $K$ (geometric).
The above-studied minimal- and maximal$-K$ models
also become reclassified as the two extreme
special cases which are, in some sense, just
most elementary.

\subsection{Outlook.}

The phenomenology-oriented core of our present message is that
one of the most promising innovative
means of the control of unitary
quantum dynamics
should be sought in a
purposeful
manipulation with the Kato's exceptional points
$g^{(EP)}$.

At the first sight such a statement sounds like an oxymoron
because
the unitarity of the evolution
requires, in Schr\"{o}dinger picture,
the self-adjointness of the Hamiltonian,
while such a requirement is
{\em manifestly\,} violated by $H(g)$ at $g=g^{(EP)}$.
For this reason it is necessary to emphasize
the existence of the two tacit assumptions behind the PHQM SP theory.
The first one is that we really {\em exclude\,}
the singularity $g=g^{(EP)}$, and that we only work in
its vicinity ${\cal D}_{(physical)}$,
with the parameter $g$
admitted to lie arbitrarily close to
$g^{(EP)}$ (i.e., formally,
$g^{(EP)} \in \partial {\cal D}_{(physical)}$).

The second tacit assumption
is more standard and means the acceptance
of the currently very popular PHQM
update of quantum theory.
In this framework
the self-adjointness of $H(g)$ is considered
$g-$dependent or, more precisely,
metric-operator-dependent, $\Theta(g)-$dependent.
Precisely due to this freedom,
the desirable
limiting transition
$g\to g^{(EP)}$ can always be performed in
a mathematically consistent and
unitarity-compatible manner.

In the current literature, unfortunately,
the
EP-related field of phenomenology
is predominantly developed in
its applications to the
open (i.e., in other words, manifestly
non-unitary) quantum systems \cite{Nimrod}
and/or to various
non-quantum
or even non-linear
systems \cite{Christodoulides,Carlbook,kon,kono}.
One of the reasons
is that in such a setup
the model-building process is
technically easier, being
still allowed to work
with the trivial choice of $\Theta(g)=I$.

In
this sense, we tried to lower here
the related psychological barriers.
Having emphasized the fundamental aspect
of the strictly unitary theory (requiring,
typically, nontrivial metrics
$\Theta(g)\neq I$),
we illustrated
its user-friendly
nature by a detailed analysis of certain
finite-dimensional $N$ by $N$ matrix benchmark
Hamiltonians $H^{(N)}(g)$.

The validity of our conclusions remains
model-independent of course.
In their brief summary let us emphasize that
one of the key benefits of the PHQM formulation
of the theory lies in its capability of covering
multiple apparently exotic system-evolution
scenarios in which $g$ is not too far from $g^{(EP)}$.
Then,
the metric $\Theta(g)$
becomes very different from the conventional
choice of $\Theta(g)=I$
of textbooks.
Besides an undeniable
phenomenological appeal of the anisotropy
$\Theta(g)\neq I$ the second deep merit of the scenario
lies in the one-to-one correspondence between
the geometry of Hilbert space and the characteristics of
the EP.
In this manner, the behavior of dynamics becomes
directly controlled by
the characteristics
of the EP,
i.e., by its algebraic multiplicity $N$ and
by its geometric multiplicity $K$.
In the vicinity of a given EP,
the latter two
integers
will
characterize the dynamics
in a unified qualitative manner tractable
as a certain quantum
analogue of the
classical Thom's catastrophe theory.

\section{Summary}

In Introduction we formulated our present project as a
transfer of the Thom's classical concept of catastrophes
to quantum theory. We reminded the readers
that the geometric nature of the Thom's theory (in
which the stability of a long-time equilibrium of the system in
question is mimicked and simulated by the local stability of a local
minimum of the so called Lyapunov function $V(x)$) cannot easily be
transferred to quantum mechanics, i.a., due to the phenomenon of
tunneling (see also \cite{catast}).

Now, let us add that there still exist multiple parallels between
the present considerations and the Thom's theory. Indeed, in the
latter case, a classification of classical catastrophes was
achieved via the reduction of arbitrary $V(x)$s to its ``canonical''
form. The resulting bifurcation scenarios were given the intuitively
appealing names (like the ``fold catastrophe'' with ``canonical''
one-parametric $V(x)=x^3+ax$, or the ``cusp catastrophe''  with  the
two-parametric but still one-dimensional $V(x)=x^4+ax^2+bx$, etc
\cite{Zeemanc}).

All this made the classical Thom's theory popular.
On this background we pointed out, in \cite{168}, that many of the
standard Lyapunov functions $V(x)$ could rather easily be
reinterpreted as mimicking certain strictly quantum analogues of the
classical elementary catastrophes. Indeed, once we decided to define
the catastrophes, qualitatively, as the ``sudden shifts in behavior
arising from small changes in circumstances'' \cite{Zeema}, we were
immediately able to reinterpret many (i.e., not all!) Lyapunov
functions $V(x)$ as the ``benchmark'' quantum potentials in
Schr\"{o}dinger Eq.~(\ref{eqsch}) with $H = -\triangle + V(x)$.

The latter idea found its constructive applications even in more
dimensions, with $x \in \mathbb{R}^d$ at nontrivial $d=2$ in
\cite{2d}, or at the more realistic $d=3$ in \cite{3d}.
Nevertheless, the price had to be paid for the strictly shared
locality of the benchmark potentials $V(x)$. This made the fairly
close classical-quantum analogy incomplete and, unfortunately, just
approximative.
Indeed, a key weakness of the approach lied in the nature
of the assignment of a suitable EP parameter $g^{(EP)}$ to the
corresponding quantum catastrophe. The reason was that
in a way motivated by the
above-cited Stone theorem, the Hamiltonians were
chosen self-adjoint. This implied
that  Im$(g^{(EP)}) \neq 0$. Thus, the
unavoidable presence of a small imaginary components in the
parameters made the process of reaching the phase transition
non-unitary. In
other words, the simulation of the quantum
``energy-level-degeneracy'' catastrophe (achieved, in our preceding
sections, due to the hidden Hermiticity of $H$) would require
an analytic continuation. Without such a
modification
of the model, the ``shifts in behavior'' would not be ``sudden'', and
the well known ``avoided level mergers''
would be experimentally observed.
In comparison, our present, EP-related simulation of the
energy-level mergers proved more successful, exact and ``unavoided''
(see also, in this context, the exactly solvable non-Hermitian
differential-operator model in \cite{ptho}).

In any case, a word of
warning must be added. The point is that the domain of the realistic
physics in which one deals with the concept of the quantum phase
transition (for reference, the readers should consult, e.g., the
Sachdev's classical monograph \cite{[a]}) is, naturally, much larger
than its EP-based subdomain as studied and clarified in our present
paper.
In parallel, the source of optimism concerning the future
developments of our present approach
could be sought in the possibility of
a partial return to the
open-system theory.
Indeed, in its more ambitious forms one could use, typically,
Lindblad operators \cite{hata} or Liouvillians
\cite{[m1],[m2],[m3]}. In such a framework, the present
classification of some of the EP-based phase-transition processes could also
be, in a next-step development, included.

Many open question emerge in such an open-system setting at present.
They are mostly connected with the specific, Liouvillean-picture-related
phenomena like quantum jumps \cite{[m1]}, in a way moving beyond the
limitations characterizing various standard quantum master equation
descriptions \cite{[c1]}. In our preceding text we only stressed
that when speaking about Hamiltonians with EPs, one usually deals
with the information about the environment which is incomplete. Now,
let us add that some more sophisticated
open-system Hamiltonians might remain Hermitian and, thus,
fundamental. In the constructions of
this type (see, e.g., \cite{[c2],[c3]})
the consistency of the theory is achieved, via
introduction of the so called Langevin force, in the Heisenberg
picture, i.e., not in the present SP framework.

\section*{Acknowledgements}

The author is grateful to Excellence project P\v{r}F UHK
2211/2022-2023 for the financial support.


\section*{Data Availability}

 No datasets were generated or analysed during the current study.

\section*{Author Contribution statement}

The author is the single author of the paper.

\section*{Competing interests}

The author declares no competing interests.

\end{document}